\documentclass[conference,a4paper]{IEEEtran}
\usepackage[utf8]{inputenc}
\usepackage[T1]{fontenc}
\usepackage{url}              
\usepackage{cite}             
\usepackage{dsfont}
\usepackage[cmex10]{amsmath}  
\allowdisplaybreaks

\usepackage{mleftright}       
\mleftright                   

\usepackage{graphicx}         
\usepackage{booktabs}         

\usepackage{amssymb}

\usepackage{bbm}

\usepackage{tikz}
\usepackage{pgfplots}
\usepgfplotslibrary{groupplots,dateplot}

\usepackage{flushend}
\usepackage{bm}
\usepackage{graphicx}         
\usepackage{booktabs}         

\usepackage{caption,subcaption}

\usetikzlibrary{shapes,arrows}
\usetikzlibrary{positioning,fit,calc}
\tikzstyle{block} = [draw, fill=white, rectangle,
minimum height=3em, minimum width=6em]
\tikzstyle{input} = [coordinate]
\tikzstyle{output} = [coordinate]

\newcommand{\E}{ \mathbb{E} } 

\let\leq\leqslant
\let\geq\geqslant

\let\phi\varphi

\newcommand{\Eps}{\mathcal{E}}
\newcommand{\X}{\mathcal{X}}
\DeclareMathOperator{\Tr}{Tr}
\newcommand{\Y}{\mathcal{Y}}
\newcommand{\Ne}{\mathcal{N}}
\newcommand{\Hil}{\mathcal{H}}
\newcommand{\ket}[1]{\left|#1\right\rangle}
\newcommand{\bra}[1]{\left\langle#1\right|}

\newcommand{\ketbra}[2]{\ket{#1}\bra{#2}}

\usepackage{comment}

\newtheorem{theorem}{Theorem}
\newtheorem{lemma}{Lemma}

\newtheorem{corollary}{Corollary}

\newtheorem{definition}{Definition}
\newtheorem{remark}{Remark}

\usepackage{todonotes}

\begin{document}



\title{Data Processing Inequalities\\for the Quantum Guesswork}

 \author{
	   \IEEEauthorblockN{Ilyass Mejdoub, Julien Béguinot, and Olivier Rioul
	   \IEEEauthorblockA{LTCI, Télécom Paris, Institut Polytechnique de Paris, France\\
						{firstname.lastname@telecom-paris.fr}}}
	 }

%

\maketitle

\begin{abstract} 
\emph{Quantum guesswork} has emerged as a crucial measure in assessing the distinguishability of non-orthogonal quantum states. In this paper, we review and compare previous definitions of the quantum guesswork, and establish the pre and post data processing inequalities (DPIs) that are relevant for the guessing problem. We also derive improved lower and upper bounds on the quantum guesswork in terms of the Shannon entropy and Holevo's information.
\end{abstract}



\maketitle

\section{Introduction}\label{sec:intro}

The classical guessing game is as follows. Alice has access to the realization $X=x$ of a discrete random variable $X\in\X = \{x_1, ..., x_n\}$, and Bob asks questions of the form ``is $X\!=\!x_i$?'' to Alice until the answer is ``yes''.
Both Alice and Bob know the alphabet $\X$ and the probability distribution $p_X$ of~$X$. 
Bob's \emph{guessing strategy} specifies the order of the successive guesses he makes until Alice confirms that Bob has found the correct value. 
A sound strategy never asks twice for the same value of $x$; therefore, Bob's strategy  can be seen as 
a permutation $\sigma\in\mathcal{S}_n$ of the alphabet of cardinality $n$, where $\sigma(1)$ is the first guess, $\sigma(2)$ is the second, and so on until the secret $x$ is found in at most $n$ trials.

The \emph{guesswork} is the average number of Bob's guesses under his strategy $\sigma$, given by 
\begin{equation}
    \mathcal{G}_\sigma(X) \triangleq {\sum_{i=1}^n}\, i \cdot p_X(\sigma(i)). 
\end{equation}
This quantity is minimized for Bob's optimal strategy $\sigma^*$, such that the probabilities $p_X(\sigma^*(i))$ are in descending order, yielding a so-called \emph{guessing entropy}, introduced by Massey~\cite{massey1994guessing}:
\begin{equation}
    G(X) = \min_{\sigma\in\mathcal{S}_n} \mathcal{G}_\sigma(X) =\mathcal{G}_{\sigma^*}(X)= \smash[b]{\sum_{i=1}^{n}}\,  i \cdot p^{\downarrow}_X(x_i) 
\end{equation}
where $p^{\downarrow}_X(x)$ are the probabilities $p_X(x)$ sorted in descending order.

Should Bob gain supplementary side information correlated with $X$, his resulting guesswork is, on average, lower. In the presence of some side information $Y$ (taking values in an arbitrary alphabet $\mathcal{Y}$), Bob's strategy depends on his observed value~$Y=y$:
\begin{equation}
\sigma: y \in \mathcal{Y} \mapsto \sigma_y \in \mathcal{S}_n. 
\end{equation}
and the corresponding guesswork given observation $Y=y$ is
\begin{equation}
\mathcal{G}_{\sigma}(X | Y = y) = \sum_{i=1}^n\, i \cdot p_{X|Y}(\sigma_y(i)|y). 
\end{equation}
Bob's \emph{conditional guesswork} is the average guesswork over all possible observations: 
\begin{equation}
\mathcal{G}_{\sigma}(X | Y) = \E_Y \sum_{i=1}^n\, i\cdot  p_{X|Y}(\sigma_Y(i)|Y). 
\end{equation}
Now for a given strategy, all observations $y\in\mathcal{Y}$ such that $\sigma_y=\sigma$ yield the same permutation $\sigma$, and regrouping terms in the sum gives
\begin{equation}\label{eq:classicalguess}
\begin{aligned}
\mathcal{G}_{\sigma}(X | Y) &= \sum_{y\in\mathcal{Y}} \sum_{i=1}^n\, i\cdot  p_{X,Y}(\sigma_y(i),y)\\
&=\sum_{\sigma\in\mathcal{S}_n} \sum_{i=1}^n\, i\cdot  p_{X,\sigma_Y}(\sigma(i),\sigma)
\\&=
\mathcal{G}_{\sigma}(X | \sigma_Y).
\end{aligned}
\end{equation}
This means that for a fixed strategy $\sigma: y \in \mathcal{Y} \mapsto \sigma_y \in \mathcal{S}_n$,
the conditional guesswork depends on the observation $Y$ \emph{only through $\sigma_Y$}. 
%
%
As shown by Arikan~\cite{arikan1996inequality} in 1996, the optimal strategy $\sigma^*_y$ still consists in guessing in descending order of conditional probabilities $p_{X|Y}(x|y)$, and the  \emph{conditional guessing entropy} is the average guesswork for such  an optimal strategy $\sigma^*$:
\begin{equation}
G(X|Y) \!=\! \min_\sigma \mathcal{G}_\sigma(X|Y) \!=\!\mathcal{G}_{\sigma^*}(X|Y) \!=\! \E_Y\! \sum_{i=1}^{n}\,i\cdot p^{\downarrow}_{X|Y}(x_i|y) 
\end{equation}
where $p^{\downarrow}_{X|Y}(x|y)$ are the probabilities $p_{X|Y}(x|y)$ sorted in descending order.

If we now substitute the classical random variable $Y$ with quantum side information, we get the notion of \emph{quantum guesswork} introduced by Chen \textit{et al.}~\cite{chen2014minimum} as an optimization criterion  for discriminating non-orthogonal quantum states. This constitutes an alternative to the conventional error probability criterion~\cite{helstrom1969quantum, chefles2000quantum, barnett2009quantum,davies1978information,eldar2001quantum, barnett2001minimum}.
Furthermore, as
mentioned in~\cite[\S~IX]{hanson2021guesswork}, quantum guesswork can serve
as a security criterion in some quantum information processing
tasks such as certifying the safety of an exchanged key in an
imperfect implementation of a QKD protocol.

In this quantum guessing game,
Alice still has a random variable $X$ whose probability distribution $p_X$ and alphabet $\X$ are known to both parties. Now, when Alice observes the realization $X=x$, she sends Bob a quantum state $\rho_x$. 
The quantum ensemble 
\begin{equation}
\Eps \triangleq \{ (\rho_x, p_X(x))\}_{x \in \X} 
\end{equation}
 is known to both Alice and Bob. 
 The quantum states $\{\rho_x\}_{x \in \X}$ can be either pure or mixed states and  belong to the set of density matrices in some finite dimensional Hilbert space $\Hil$. 

Once again, Bob can only ask questions of the form ``is $X\!=\!x$?''. 
The only way for Bob to get any helpful information from the received state $\rho_x$ is to eventually measure it. Bob can freely manipulate the state $\rho_x$, including applying unitary transformations or conducting measurements, to correctly infer the realization~$x$.  
The most general form of quantum measurement~\cite{wilde2013quantum,watrous2018theory} is called POVM (positive operator valued measure) and is mathematically defined as a set $ E = \{E_y\}_{y\in \Y}$ of positive semi-definite operators which sum to the identity:
\begin{equation}
    E_y \geq 0 \text{ and } \sum_{y\in \Y} E_y = I.
\end{equation}
The POVM and the quantum ensemble used by Bob define a classical random variable $Y_E$ taking values in $\Y$.
By Born's rule~\cite{wilde2013quantum,watrous2018theory}:
\begin{equation}
p_{Y_E|X}(y|x)=  \Tr[E_y\rho_x].
\end{equation}
so that 
$p_{Y_E}(y) =
\sum_{x \in \X}  \Tr[E_y\rho_x]\cdot p_X(x)$.
While the channel $X \to \rho_X$ is quantum, for a fixed POVM $E$ the composition of channels 
$$
    X \to \boxed{\text{Quantum Channel}} \to \rho_X \to \boxed{\mathrm{POVM}} \to Y_E
$$
forms a classical channel. Hence, subject to the POVM $E$, we can define the quantum conditional guesswork for a given strategy $\sigma$ as a classical conditional guesswork~\eqref{eq:classicalguess}  
\begin{equation}
\mathcal{G}_{\sigma}(X | Y_{E}) 
=\mathcal{G}_{\sigma}(X | \sigma_{Y_E}) 
\end{equation}
The guesswork depends on $\rho_x$ only through $Y_E$, and as above, it depends on $Y_E$ only through $\sigma_{Y_E}$, according to the following Markov chain 
$$
    X \to \rho_X \to \boxed{\mathrm{POVM}} \to Y_E \to  \sigma_{Y_E}.
$$
The quantum guessing entropy is now the minimum average guesswork, achieved for the optimal guess strategy \emph{and} the optimal \emph{measurement} strategy (through the choice of the POVM). Therefore, a natural definition is the following:
\begin{definition}[Quantum Guesswork]\label{def:one}
The \emph{quantum conditional guessing entropy} of $X$ (also named \emph{quantum guesswork}) given a quantum ensemble $\mathcal{E}$ is defined by
 $$
    G(X|\mathcal{E}) = \min_{\sigma,E} \mathcal{G}_{\sigma}(X | Y_{E})
$$
\end{definition}
The optimization is over the pair $(\sigma,E)$; however, it only depends on $\sigma_{Y_E}$ so that the procedure can be simplified.

%
In the present literature on quantum guessing, the quantum guesswork $G(X|\Eps)$ was defined in two ways, either by minimizing a classical guessing entropy over the choice of a POVM~$E$~\cite{chen2014minimum,hanson2021guesswork,katariya2023experimental}, or through it ``operational'' point of view advocated in~\cite{dall2022guesswork,avirmed2023adversarial,dall2023measurement,dall2023computing}. 
The equivalence between these two definitions of quantum guesswork is perhaps  known in certain circles but, to the best of our knowledge, has not yet been explicitly proved.


The remainder of this paper is organized as follows.
In Section \ref{Definitions}, we restate the two previous definitions of quantum guesswork and formally show that they are equivalent to our Definition~\ref{def:one}. In addition, we review and explain anew known results about the measurement attaining the quantum guesswork. The main results of this paper are in Section \ref{Main Results}. We first establish the unitary invariance  property (\S~\ref{Unitary invariance}).  We then generalize two data processing inequalities to the quantum setting (pre-DPI in \S~\ref{Pre DPI} and post-DPI in \S~\ref{Post DPI}). Finally, in  \S~\ref{Sharpening}, we propose a general methodology to obtain lower and upper bounds on quantum guesswork in terms of the Shannon entropy and Holevo's information that improves the bound of Chen \textit{et al.}~\cite{chen2014minimum}.

\section{Definitions}
\label{Definitions}


\subsection{Notations}
\begin{itemize}
    \item $\Hil$ denotes an arbitrary finite dimensional Hilbert space;
    \item 
    $\mathcal{L}(\Hil)$ denotes the space of linear operators from $\Hil$ to $\Hil$; 
    \item ${\cal S}_n$ denotes the set of permutations on $\{1,...,n\}$;
    \item For a permutation $\sigma \in {\cal S}_n$, $\overline{\sigma}(i) \triangleq \sigma(n+1-i)$;
    \item For a normal matrix $A\in \mathcal{L}(\Hil)$ whose spectral decomposition is $A = \sum_{k=1}^d \lambda_k \ketbra{\psi_k}{\psi_k}$, we use the notation $f(A) \triangleq \sum_{k=1}^d f(\lambda_k) \ketbra{\psi_k}{\psi_k}$. In particular, 
\begin{equation}
\begin{aligned}
    |A| &\triangleq \sum_{k=1}^d |\lambda_k| \ketbra{\psi_k}{\psi_k},\\
    \mathbbm{1}_{A\in S} &\triangleq \sum_{k;\lambda_k\in S}  \ketbra{\psi_k}{\psi_k}.
\end{aligned} 
\end{equation}
    \item \emph{Loewner order}: For Hermitian matrices $A,B$, we write $A\geq B$ if $A-B$ is positive;
    \item $||.||$ denotes the trace norm $||A|| \triangleq \Tr(|A|)$;
\end{itemize}

\subsection{Quantum Guesswork: Chen et al.'s Definition $G_1(X|\Eps)$}

The first previous definition of the quantum guesswork $G_1(X|\Eps)$ given in~\cite{chen2014minimum,hanson2021guesswork,katariya2023experimental} is as follows.
From the joint distribution $p_X(x)p_{Y_E|X}(y|x)$ of $X$ and $Y_E$, we compute the (classical) \emph{guessing entropy} $G(X|Y_E)$ from the POVM $E$ and the ensemble $\Eps$. 
Bob can now only control the choice of the POVM to minimize his average number of guesses. Thus the quantum guesswork is here defined as 
\begin{equation}
    G_1(X|\Eps) \triangleq \inf_{E}G(X|Y_E).
\end{equation}

As shown in \cite[Theorem~1]{hanson2021guesswork}, the most general quantum measurement strategy Bob can use is either one of the following strategies:
\begin{itemize}
\item perform a single measurement using an arbitrary POVM and then apply the classical guessing strategy;
\item perform a single measurement using a POVM indexed by permutations $\sigma\in {\cal S}_n$ specifying the order of the successive guesses;
\item perform successive measurements using POVMs indexed by the remaining possible values using Alice's feedback to each of Bob's guesses.
\end{itemize}

\subsection{Quantum\,Guesswork: Dall’Arno\,et\,al.'s\,Definition\,$G_2(X|\Eps)$}

The second previous definition of the quantum guesswork $G_2(X|\Eps)$ given in~\cite{dall2022guesswork,avirmed2023adversarial,dall2023measurement,dall2023computing} is as follows.
Without loss of generality, we can assume $\X = \{1, ..., n\}$, and we restrict ourselves to special types of POVMs ${E}$ indexed by \emph{permutations}:
\begin{equation}
 {E}= \{{E}_{\sigma}\}_{\sigma \in {\cal S}_n}
\end{equation}
whose outcome is \emph{operational} in the sense that it gives Bob directly its entire ranking of guesses in the form of a permutation $\sigma \in {\cal S}_n$ where $\sigma(i)$ is Bob's $i$-th guess. 

The probability that the outcome of Bob's measurement is $\sigma$ and that his $i$-th guess is the correct one (i.e $\sigma(i) = x$) given that he received the state $\rho_x$ is obtained by Born's rule
\begin{equation}
    p(\sigma, i) = \Tr[E_{\sigma}\cdot \rho_{\sigma(i)}]\cdot p_X(\sigma(i)).
\end{equation}
If we marginalize over $\sigma$ we get the probability that the $i$-th guess is correct given that the received state is $\rho_x$,
\begin{equation}
    p(i) = \sum_{\sigma \in {\cal S}_n}\Tr[E_{\sigma}\cdot \rho_{\sigma(i)}]\cdot p_X(\sigma(i)).
\end{equation}
Bob's guesswork for ensemble $\Eps$ and POVM $E$ is now 
\begin{equation}
\mathcal{G}(\Eps, E)\triangleq \sum_{i=1}^{n} i\cdot p(i),
\end{equation}
where the values $p(i)$ are not necessarily in decreasing order, and the quantum guesswork is defined as its minimum  over all POVMs ${E} = \{{E}_{\sigma}\}_{\sigma \in {\cal S}_n}$:
\begin{equation}
 G_2(X|\Eps) \triangleq \inf_{{E}} \mathcal{G}(\Eps, E). 
\end{equation}

\subsection{Equivalence Between Definitions $G_1(X|\Eps)$ and $G_2(X|\Eps)$ }

\begin{theorem}
Definitions of Chen \emph{et al.} and Dall’Arno \emph{et al.} are equivalent to our Definition~\ref{def:one}:
\begin{equation}
    G_1(X|\Eps) = G_2(X|\Eps) =G(X|\Eps).
\end{equation}
\end{theorem}

\begin{proof}
By definition of conditional guessing entropy,
\begin{equation}
G(X|Y_E)=\min_\sigma \mathcal{G}_{\sigma}(X|{Y_E}) \!=\!\mathcal{G}_{\sigma^*}(X|{Y_E}) 
\end{equation}
where the optimal strategy $\sigma^*$ is such that for each $y$, the probabilities $p_{X|Y_E}(\sigma^*_Y(i)|y)$ are sorted in descending order.
It follows that
\begin{equation}
    G(X|\mathcal{E})  = \min_E\min_{\sigma} \mathcal{G}_{\sigma}(X | {Y_E}) 
    = \min_E G(X|Y_E)=G_1(X|\mathcal{E}).
\end{equation}
Now let $\sigma: y\in\mathcal{Y}\mapsto \sigma_y$ be Bob's (not necessarily optimal) strategy and recall the Markov chain
$$
 \rho_x \to \boxed{\mathrm{POVM}} \to Y_E \to \sigma_{Y_E}
$$
and the definition
\begin{equation}
\mathcal{G}_\sigma (X|Y_E) =\sum_y \sum_i\, i\cdot p_{X,Y_E}(\sigma_y(i),y)
\end{equation}
where by Born's rule,
\begin{equation}
 p_{X,Y_E}(x,y)=p_X(x)\cdot \Tr[E_y\rho_x].
\end{equation}
Again grouping terms in $y$ corresponding to the same $\sigma_y=\sigma$, we have
\begin{equation}
  p_{X,\sigma_{Y_E}}(x,\sigma)=p_X(x) \sum_{y\mid \sigma_y=\sigma}\Tr[E_y\rho_x] = 
  p_X(x) \Tr[\tilde{E}_\sigma\rho_x]
\end{equation}
where
\begin{equation}
    \tilde{E}_\sigma = \sum_{y\mid \sigma_y=\sigma } E_y. 
\end{equation}
is another POVM $\widetilde{E}$, indexed by permutations $\sigma\in\mathcal{S}_n$, since 
$\tilde{E}_\sigma\geq 0$ and $\sum_\sigma \tilde{E}_\sigma=I$.
The output of a measurement is now a realization of a random permutation $\Sigma\triangleq Y_{\tilde{E}}\in\mathcal{S}_n$, which directly gives Bob's strategy, where
\begin{equation}
\begin{aligned}
\mathcal{G}_\sigma (X|Y_E)&=\mathcal{G}_{\mathrm{Id}} (X|\Sigma) 
\\&=\sum_\sigma \sum_i\, i\cdot p_{X,\Sigma}(\sigma(i),\sigma)\\
&= \sum_\sigma \sum_i\, i\cdot p(i,\sigma)= \sum_i\, i\cdot p(i)
\end{aligned}
\end{equation}
where $\mathrm{Id}$ is the identity permutation since by construction, $\Sigma=\sigma\iff \sigma_{y_{\tilde{E}}}=\sigma$.
Thus, minimizing over $\tilde{E}$ is the same as minimizing over both $E$ and $\sigma$.
Therefore,
\begin{equation}
G(X|\mathcal{E})  = \min_{E}\min_\sigma \mathcal{G}_{\sigma}(X | Y_E) = \min_{\tilde{E}}\mathcal{G}_{\mathrm{Id}} (X|\Sigma)  = G_2(X|\mathcal{E})
\end{equation}
which ends the proof.
\end{proof}

\subsection{Remarks on Attaining the Minimum}
\label{attaining the min}

Dall'Arno et al.~\cite{dall2022guesswork} found a closed form expression $G(X|\Eps)$ under some conditions on the quantum ensemble $\Eps$, using his definition $G_2(X|\Eps)$ above. 
For convenience write $E$ instead of $\tilde{E}$, and let
\begin{equation}
  \mathcal{G}(\Eps,E) \triangleq \mathcal{G}_{\mathrm{Id}} (X|\Sigma)
\end{equation}
to stress the dependence on the quantum ensemble $\Eps$ and the POVM $E$.
With that definition, the objective function to be minimized over the space of POVMs is
\begin{align}
  \mathcal{G}(\Eps,E) &= \sum_{i=1}^n i\cdot \sum_{\sigma \in {\cal S}_n}\Tr[E_{\sigma}\cdot \rho_{\sigma(i)}]\cdot p_X(\sigma(i)) \\
    &= \sum_{\sigma} \Tr[E_\sigma \cdot f_{\sigma}]
\end{align}
where
$f_{\sigma} \triangleq \sum_{i=1}^n i\cdot p_X(\sigma(i)) \rho_{\sigma(i)}$.
To simplify this expression consider its evaluation on $\overline{\sigma}$:
\begin{equation}
   f_{\overline{\sigma}} = (n+1)\widehat{\rho} - f_{\sigma} 
\end{equation}
where
$
\widehat{\rho}  \triangleq \sum_{i=1}^n \rho_{\sigma(i)}p_X(\sigma(i))
$ does not depend on $\sigma$.
Therefore, defining 
\begin{equation}
\begin{aligned}
        \Eps_{\sigma} &\triangleq 2f_{\sigma} - (n+1)\widehat{\rho}\\
        &= \sum_{i=1}^n (2t-n-1) p_X(\sigma(i))\cdot \rho_{\sigma(i)}.
\end{aligned}
\end{equation}
we have
\begin{equation}
\begin{aligned}
    \mathcal{G}(\Eps,E)  &= \dfrac{n+1}{2} + \dfrac{1}{2}\sum_{\sigma} \Tr[E_{\sigma}\cdot \Eps_{\sigma}]\\
    &= \dfrac{n+1}{2} - \dfrac{1}{2}\sum_{\sigma} \Tr[E_{\sigma}\cdot \Eps_{\overline{\sigma}}]\\
\end{aligned}
\end{equation}
since
$\Eps_{\overline{\sigma}} = -\Eps_{\sigma}$.

We need the following lemma.
\begin{lemma}
\label{lemma-trace}
    Let $A$ be an Hermitian positive matrix. If  $X$ and $Y$ are Hermitian matrices such that $X \leq Y$, then $\Tr[AX] \leq \Tr[AY]$.
\end{lemma}
\begin{proof}
Let $Y-X=\sum_{i=1}^d \lambda_i \ketbra{\phi_i}{\phi_i}$ be the spectral decomposition of $Y-X\geq 0$. Then
    \begin{align}
        \Tr[AY] - \Tr[AX] &= \Tr[A(Y-X)] \\
                          &= \Tr[A\sum_{i=1}^d \lambda_i \ketbra{\phi_i}{\phi_i}] \\
                          &= \sum_{i=1}^d \lambda_i \smash[b]{\Tr[\underbrace{\bra{\phi_i}A\ket{\phi_i}}_{\geq 0}] }
                          \geq 0 \hspace*{2cm}\IEEEQEDhere\notag
    \end{align}
\end{proof}

Applying Lemma~\ref{lemma-trace} to $|\Eps_{\overline{\sigma}}| \geq \Eps_{\overline{\sigma}}$ we obtain
\begin{equation}
   \mathcal{G}(\Eps,E)
    \geq \dfrac{n+1}{2} - \dfrac{1}{2}\sum_{\sigma} \Tr[E_{\sigma}\cdot |\Eps_{\overline{\sigma}}|]\\
\end{equation}
where we have used~Lemma~\ref{lemma-trace}.
Now, if there exists $\sigma^* \in {\cal S}_n$ such that $|\Eps_{\sigma^*}| \geq |\Eps_{\sigma}|$ for all $\sigma \in {\cal S}_n$,
then
\begin{equation}
\begin{aligned}
   \mathcal{G}(\Eps,E)
    &\geq \dfrac{n+1}{2} - \dfrac{1}{2}\sum_{\sigma} \Tr[E_{\sigma}\cdot |\Eps_{\sigma^*}|]\\
    &=\dfrac{n+1}{2} - \frac{1}{2} \Tr[| \Eps_{\sigma^*}|]
\end{aligned}
\end{equation}
since $\sum_\sigma E_\sigma=I$. Now
\begin{equation}
    - | \Eps_{\sigma^*}| = 
    \mathbbm{1}_{\Eps_{\sigma^*}<0}\Eps_{\sigma^*}-\mathbbm{1}_{\Eps_{\sigma^*}>0}\Eps_{\sigma^*}
    = E^{\sigma^*}_{\sigma^*} \Eps_{\sigma^*}
    +E^{\sigma^*}_{\overline{\sigma^*}} \Eps_{\overline{\sigma^*}}
\end{equation}
where $E^{\sigma^*}$ is the POVM defined by
\begin{equation}
E^{\sigma^*}_{\sigma} = 
\begin{cases} 
\mathbbm{1}_{\Eps_{\sigma}<0}+
 \frac{1}{2}\mathbbm{1}_{\Eps_{\sigma}=0} & \text{if } \sigma \in \{\sigma^*, \overline{\sigma^*}\}, \\
0 & \text{otherwise}.
\end{cases}
\end{equation}
Indeed, one easily checks that $E^{\sigma^*}_{\sigma}\geq0$, $\sum_\sigma E^{\sigma^*}_{\sigma} = E^{\sigma^*}_{\sigma^*}+E^{\sigma^*}_{\overline{\sigma^*}}=\mathbbm{1}_{\Eps_{\sigma^*}<0}+
 \frac{1}{2}\mathbbm{1}_{\Eps_{\sigma^*}=0} +
 \mathbbm{1}_{\Eps_{\overline{\sigma^*}}<0}+
 \frac{1}{2}\mathbbm{1}_{\Eps_{\overline{\sigma^*}}=0} 
 =\mathbbm{1}_{\Eps_{\sigma^*}>0} +
 \mathbbm{1}_{\Eps_{{\sigma^*}}<0}+
 \mathbbm{1}_{\Eps_{{\sigma^*}}=0} =I
$.
It follows that 
\begin{equation}
      \mathcal{G}(\Eps,E)
    \geq  \dfrac{n+1}{2} + \dfrac{1}{2}\sum_{\sigma} \Tr[E^{\sigma^*}_{\sigma}\cdot \Eps_{\sigma}]   
\end{equation}
Hence $   \mathcal{G}(X | \Sigma) $ is minimized for $E=E^{\sigma^*}$.
This explains why this particular POVM achieves the minimum guessing work Dall'Arno et al.'s result:
\begin{theorem}[{\cite[Theorem~1]{dall2023measurement}}]
\label{theorem dall arno}
    if $\exists \sigma^* \in {\cal S}_n$ s.t $\forall \sigma \in {\cal S}_n$ $|\Eps_{\sigma^*}| \geq |\Eps_{\sigma}|$ then
    \begin{equation}
        G(X|\Eps) =   \dfrac{n+1}{2} + \dfrac{1}{2}\sum_{\sigma} \Tr[E^{\sigma^*}_{\sigma}\cdot \Eps_{\sigma}]   
    \end{equation}
\end{theorem}
This closed form formula is only valid under the condition $|\Eps_{\sigma^*}| \geq |\Eps_{\sigma}|$. In particular, every ensemble made of uniformly distributed qubits satisfies this hypothesis~\cite{dall2023measurement}. 
However since the Loewner order $A\geq B$ is not total, the permutation $\sigma^*$ in Theorem \ref{theorem dall arno} does not always exist, and the problem of efficient computation of the quantum guesswork is still open for many classes of ensembles.

\section{Data Processing Inequalities} 
\label{Main Results}

\subsection{Unitary Invariance}
\label{Unitary invariance}
In the framework of the guessing game, if a unitary transformation is applied to the quantum ensemble, we expect the value of quantum guesswork to remain the same:    
\begin{theorem}
\label{U-invariance}
Let $U$ be a unitary matrix and define
    \begin{equation}
        U\Eps U^\dagger = \{U\rho_xU^\dagger, p_X(x)\}_{x \in X}.
    \end{equation}
Then 
    \begin{equation}
        G(X|\Eps) = G(X|U\Eps U^\dagger).
    \end{equation}
\end{theorem}

\begin{proof}
First observe that $\Tilde{E}_{y} \triangleq U^\dagger E_{y}U$ forms a POVM since 
$\sum_y \Tilde{E}_{y} = U^\dagger \sum_y E_y U = U^\dagger U = I$.
From the identity
\begin{equation}
\Tr[E_y\cdot U\rho_iU^\dagger]=\Tr[U^\dagger E_yU\cdot \rho_i],
\end{equation}
 it is easily seen that   
    \begin{equation}
        \mathcal{G}(U\Eps U^\dagger, E) 
        = \mathcal{G}(\Eps, U^\dagger EU).
    \end{equation}
Therefore, since the map $E \mapsto \tilde{E}=U^\dagger EU$ is one-to-one,
\begin{align}
        G(X|U\Eps U^\dagger) &= \inf_{E}\mathcal{G}(U\Eps U^\dagger, E) \\
                   &= \inf_{E}\mathcal{G}(\Eps, U^\dagger EU) \\
                   &= \inf_{\tilde{E}}\mathcal{G}(\Eps, \tilde{E}) \\ 
                   &= G(X|\Eps)\hspace*{5cm}\IEEEQEDhere\notag
\end{align}
\end{proof}

\subsection{Post Data Processing Inequality}
\label{Post DPI}

More generally, assume that a quantum channel is applied to the quantum ensemble.
We expect that guessing the label of the output of the quantum channel is more difficult than directly guessing the initial state. This is a generalization of the \emph{data processing inequality} for guesswork $G(X|Y)$ to the quantum case:
\begin{theorem}\label{thm-postdpi}
Let $\Ne$ be a quantum channel (a completely positive trace-preserving map) and define
\begin{equation}
\Ne(\Eps) = \{\Ne(\rho_x), p_X(x)\}_{x \in \X}.  
\end{equation}
Then
    \begin{equation}
        G(X|\Eps) \leq G(X|\Ne(\Eps))
    \end{equation}
\end{theorem}
\begin{proof}
Let 
    \begin{equation}
        \Ne(\rho_x) = \sum_{l=1}^d V_l\rho_xV_l^\dagger
    \end{equation}
be the operator-sum representation~\cite[Theorem~4.4.1]{wilde2013quantum} of $\Ne$, where 
$\sum_{l=1}^d V_l^\dagger V_l = I$.
Since 
\begin{equation}
\Tr[E_y\cdot \Ne(\rho_{x_i})] =  \Tr[E_y\cdot \sum_{l=1}^d V_l\rho_{x_i}V_l^\dagger]
=\Tr[\sum_{l=1}^d V_l^\dagger E_yV_l\cdot \rho_{x_i}],
\end{equation}
one has
    \begin{equation}
        \mathcal{G}(\Ne(\Eps), E) 
= \mathcal{G}(\Eps, \widetilde{\Ne}(E))
    \end{equation}
    where $\widetilde{\Ne}(E_y) = \sum_{l=1}^d V_l^\dagger E_yV_l\geq 0$ is a POVM since 
   \begin{equation}
        \sum_{y} \widetilde{\Ne}(E_y) = \sum_{l=1}^d V_l^\dagger {\sum_yE_y}       V_l
        =  \sum_{l=1}^d V_l^\dagger V_l = I.
    \end{equation}
Therefore,
    \begin{align}
        G(X|\Ne(\Eps)) &= \inf_{E} \mathcal{G}(\Ne(\Eps), E) 
                       = \inf_{E} \mathcal{G}(\Eps, \widetilde{\Ne}(E)) \\
                       &\geq \inf_{E} \mathcal{G}(\Eps, E) \\
                       &= G(X|\Eps) \hspace*{5.1cm}\IEEEQEDhere\notag
    \end{align} 
\end{proof}

\begin{remark}
Note that equality does not hold in general since $E \mapsto \widetilde{\Ne}(E)$ is not always one-to-one.
However, one recovers Theorem~\ref{U-invariance} by applying Theorem~\ref{thm-postdpi} twice, once for the channel $\rho\mapsto U\rho\, U^\dagger$ and once for the inverse channel $\rho\mapsto U^\dagger\rho\, U$.
\end{remark}

\subsection{Pre Data Processing Inequality}
\label{Pre DPI}

Here we use the concept of \emph{majorization}~\cite{majorization}.
We say that a random variable $Z$ majorizes $X$ if it distribution $p_Z$ majorizes $p_X$ and we write $X \preceq Z$.
\begin{theorem}
    $G(X|\Eps)$ is Schur-concave, i.e.,
    \begin{equation}
        X \preceq Z \implies G(X|\Eps) \geq G(Z|\Eps).
    \end{equation}
\end{theorem}
\begin{proof}
    By Schur-concavity of guessing entropy and conditional guessing entropy~\cite[Theorem~5]{beguinot2022my},~\cite{rioulinterplay}, we get
    \begin{align}
    X \preceq Z &\implies G(Z|Y_E) \leq G(X|Y_E)\\
                   &\implies \inf_{E}G(Z|Y_E) \leq \inf_{E}G(X|Y_E)\\
                   &\implies G(Z|\Eps) \leq G(X|\Eps).\hspace*{3.7cm}\IEEEQEDhere\notag
    \end{align} 
\end{proof}
Since $X$ majorizes $f(X)$~\cite{rioulinterplay}, we have the following
\begin{corollary}
    For a deterministic function $f$, we have
    \begin{equation}
        G(f(X)|\Eps) \leq G(X|\Eps)
    \end{equation}
\end{corollary}
Guessing a function of $X$ (e.g., its least or most significant bit) is obviously easier than guessing $X$ directly.

\subsection{From Classical to Quantum Guessing vs. Entropy}
\label{Sharpening}
Closed-form solutions to the quantum guesswork problem are only known for some specific ensembles under certain conditions (geometrically uniform states \cite[Section~5]{chen2014minimum}, qubits under a uniform distribution \cite[Corollary~1]{dall2022guesswork}\cite[Section~B]{dall2023measurement}). When closed-form solutions are not available, it may be desirable to obtain tight upper and lower bounds in terms of the Shannon entropy of $X$ and quantities such as Holevo's information, which can be computed from the quantum ensemble $\Eps$.
In the following, we generalize the approach of \cite{chen2014minimum} to obtain lower and upper bounds from any type of classical inequalities between $H(X|Y)$ and $G(X|Y)$.

Let $X$ and $Y_E$ be two random variables defined as in~\S~\ref{sec:intro}. \emph{Accessible information}, denoted $I_{acc}(\Eps)$, is the maximum amount of information that Bob can extract from his received state $\rho_x$ about $X$, defined as
\begin{equation}
    I_{acc}(\Eps) \triangleq \sup_{E} I(X:Y_E).
\end{equation}
There is no known formula to analytically solve the optimization problem in the definition of accessible information, but we have the following tight upper \cite{Holevo1973StatisticalPI} and lower \cite{jozsa1994lower} bounds:
\begin{equation}\label{bounds accessible}
    \Lambda(\Eps) \leq I_{acc}(\Eps) \leq \chi(\Eps).
\end{equation}
Here $\chi(\Eps)$ (Holevo's information) and $\Lambda(\Eps)$ are defined by
\begin{align}
    \chi(\Eps) &= S\bigl(\sum _{x}p_X(x)\rho_{x}\bigr)-\sum_{x}p_X(x)S(\rho_{x}) \\
    \Lambda(\Eps) &= Q\bigl(\sum_{x}p_X(x)\rho_x\bigr) - \sum_{x}p_X(x)Q(\rho_x)
\end{align}
where $S$ denotes the Von Neumann entropy and $Q$ denotes the subentropy, a quantity introduced in~\cite{jozsa1994lower}:
\begin{align}
    S(\rho) &= -\sum_k \lambda_k\log(\lambda_k) \\
    Q(\rho) &= -\sum_k\prod_{l \neq k}\dfrac{\lambda_k}{\lambda_k-\lambda_l}\lambda_k\log({\lambda_k})
\end{align}
where the $\lambda_k$'s are the eigenvalues of $\rho$.




\begin{theorem}\label{thm:guessing-inequality}
Let $\phi$ and $\psi$ be two non-decreasing functions such that for Shannon's entropy and the classical guessing entropy, 
\begin{equation}
    \psi(H(X|Y)) \geq G(X|Y) \geq \phi( H(X|Y) ).
\end{equation}
Then  for the quantum guesswork, we have
\begin{equation}
    \psi(H(X)-\Lambda(\Eps)) \geq G(X|\Eps) \geq \phi( H(X)-\chi(\Eps) ).
\end{equation}
\end{theorem}

\begin{proof}
    For any POVM $E$, we have
    \begin{equation}\label{above}
        \psi(H(X|Y_E)) \geq G(X|Y_E) \geq \phi(H(X|Y_E)).
    \end{equation}
    Now, using the bounds in \eqref{bounds accessible}, 
    \begin{equation}
        H(X) - \Lambda(\Eps) \geq \inf_{E}H(X|Y_E) \geq H(X) - \chi(\Eps).
    \end{equation}
    Taking the infimum over $E$ in \eqref{above} we obtain
    \begin{equation}
        \psi(H(X) - \Lambda(\Eps)) \geq G(X|\Eps) \geq \phi(H(X) - \chi(\Eps)).
    \end{equation}
    since $\phi$ and $\psi$ are non-decreasing.
\end{proof}

Using the classical McEliece-Yu inequality~\cite{McEliece}
we recover the upper bound on the quantum guesswork obtained in~\cite[Thm~3]{chen2014minimum}. 
In this case McEliece-Yu inequality cannot be improved since it is achieved everywhere when the channel $X\to Y$ is the $n$-ary erasure channel.

Using the classical Massey inequality~\cite{massey1994guessing} we recover~\cite[Thm~2]{chen2014minimum}.
With Rioul's inequality~\cite[Eq.~13]{rioul2022variations}, we obtain the improved inequality
\begin{equation}
    G(X|\Eps) \geq
     \dfrac{2^{H(X)-\chi(\Eps)}}{e} +\dfrac{1}{2}.
\end{equation}
Since the above Theorem is generic it can also be used with other improved  bounds, e.g., those obtained in~\cite{DBLP:conf/isit/BeguinotR24}.

\section*{Acknowledgments}
The authors are thankful to Michele Dall’Arno for the numerous exchanged emails and his help in understanding the basic theory of quantum guesswork.


\bibliographystyle{IEEEtran} 
\bibliography{paper}

\end{document}